\documentclass[letterpaper,12pt]{article}

\usepackage{amsmath}
\usepackage{amsthm}
\usepackage{amssymb}
\usepackage{delarray}
\usepackage[dvips]{graphics}
\usepackage{epsfig}
\usepackage{color}

\newcounter{subclaim}

\newtheorem{thm}{Theorem}
\newtheorem{lemma}[thm]{Lemma}

\theoremstyle{definition}
\newtheorem{defn}{Definition}

\theoremstyle{remark}

\title{Differential and graphical approaches to multistability theory for chemical reaction networks}
\author{ Mark Lipson \thanks{Written while affiliated with Department of Mathematics, Harvard University.  Email address: mark.lipson@post.harvard.edu} }

\begin{document}
\maketitle

\begin{abstract}
The use of mathematical models has helped to shed light on countless phenomena in chemistry and biology.  Often, though, one finds that systems of interest in these fields are dauntingly complex.  In this paper, we attempt to synthesize and expand upon the body of mathematical results pertaining to the theory of multiple equilibria in chemical reaction networks (CRNs), which has yielded surprising insights with minimal computational effort. Our central focus is a recent, cycle-based theorem by Gheorghe Craciun and Martin Feinberg, which is of significant interest in its own right and also serves, in a somewhat restated form, as the basis for a number of fruitful connections among related results.
\end{abstract}

\section{Introduction}

Chemical reactions have long been a popular and fruitful subject for exploration with mathematical models.  A wide range of behaviors can emerge from the evolution of a chemical system over time, and these can frequently be explained in a useful way by analyzing a dynamical model.  Even in complex biological systems, one often finds it possible to isolate tractable networks of interacting proteins and to shed light on certain of their fundamental properties using fairly simple chemistry and mathematics.

In this paper, we will focus on the capacity or incapacity of chemical reaction networks (CRNs) to admit multiple equilibria in their dynamics.  Naturally, the topic of multistability (or its opposite, injectivity, depending on one's perspective) has attracted interest over the years in the chemistry, physical chemistry, and chemical engineering literature \cite{Handwave, ClarkeLong, ClarkeShort, CF1, CF2, CF3, Feinberg1, Feinberg2, Hunt, SCL}.  More recently, biologists have also come to recognize the importance of multistability for living systems, especially on the molecular scale, where the presence of multiple equilibria can be associated with switch-like processes and biological memory, particularly in cell differentiation and development.  Specific areas of study have included gene networks \cite{switch, Toggle, Bose, Soule2}, enzyme catalysis and induction \cite{CF2, CFPNAS, Degn, LAC}, and cellular signaling \cite{Bhalla, FX, JG1, HuangFerrell, Immune, Markevich, Pom}.  A number of the results we will discuss have also come from farther afield, with applications ranging from ecology~\cite{Eco} to economics \cite{BMQ, GN, Parth, Sam} to probability theory \cite{Parth}.



In describing a CRN as a dynamical system, we will be dealing with the ways in which some state variables - the concentrations of chemical species in the system - influence the rates of change of others over time.  One species can activate or inhibit a second species, which can in turn have an effect on a third, and so on.  In the study of multistability, we will be particularly concerned with the (usually indirect) influence that one state variable exerts on its own value.  This kind of \textquotedblleft cycle" (as we will call it) is often easy to describe qualitatively, simply by examining the reactions in the network, and in many cases this proves to be enough to yield important and subtle information about the properties of the network \cite{CF1}.  For example, if we were to observe multistability in a chemical process composed of several unknown steps, we could apply the results in this paper to help us decide among proposed alternatives for the underlying series of steps, even if the possibilities were superficially very closely related.   


A recent theorem by Gheorghe Craciun and Martin Feinberg \cite{CF1, CF2} is the central focus of our paper.  Over the last twenty years, Feinberg has been the leading developer of \textquotedblleft chemical reaction network theory," or CRNT.   His earlier work on multistability was based on a property known as the \emph{deficiency} of a network \cite{Feinberg1, Feinberg2, Leib}, but more recently, he and his colleagues have focused on cycle-based conditions~\cite{CF1, CF2, CF3, CFPNAS, SCL}, which can be compared and combined with other similar results in a number of productive ways.\footnote{While deficiency theory is certainly a worthy subject for study, the mathematics behind it is fairly specialized, and we will not be dealing with the subject here.  Similarly, we will not delve into the stoichiometric network analysis (SNA) of Bruce Clarke \cite{Handwave, ClarkeLong, ClarkeShort}.}  Feinberg has not been widely cited in either the biology or the mathematics literature, so one of our goals is to bring attention to what has been an underappreciated body of work.


We begin with a series of definitions and a discussion of some of their implications in sections~\ref{Notation}~-~\ref{remarks}, followed by some general results in section~\ref{DS} regarding global injectivity.  Section~\ref{graphsection} introduces graphs and cycles and the Thomas-Soul\'{e} theorem, and section~\ref{CFsection} presents a generalized and substantially reformulated exposition of Craciun-Feinberg.  We conclude with a discussion of the relationships among the various Jacobian-based theorems on multistability in section~\ref{compare} and an example of the theory in action in section~\ref{MAPK}.

\section{Preliminary definitions and notation}\label{Notation}

Imagine a container, either a laboratory vessel or a living cell, within which some chemical reactions take place at constant temperature.  Let the active chemical species be $S_1, S_2, \dots, S_n$, and suppose there are $m$ reactions that occur among them.  The~$j^{th}$ reaction is given by $c_{j1}S_1 + c_{j2}S_2 + \dots + c_{jn}S_n \to d_{j1}S_1 + d_{j2}S_2 + \dots + d_{jn}S_n$, where the $c_{ji}$ and $d_{ji}$ are nonnegative integers.  Some reactions may be reversible, in which case we use the shorthand $c_{j1}S_1 + c_{j2}S_2 + \dots + c_{jn}S_n \leftrightharpoons d_{j1}S_1 + d_{j2}S_2 + \dots + d_{jn}S_n$.  For now, we will consider the forward and reverse directions to be two separate reactions.

At time $t$, we write $x_i(t)$ for the molar concentration of species $S_i$ in the container, with $x_i(t) \in \mathbb{R}^+$ for each $i$ and all $t$, and let $\mathbf{x}$ be the vector $(x_1, x_2, \dots, x_n)$.  (At times, we will also use the notation $[S_i]$ to denote the concentration of $S_i$.  These quantities change over time due to the reactions that take place, which we will assume are elementary and are thus governed by mass-action kinetics (see section~\ref{massaction}).  This means that the rate (in concentration per time) of reaction $j$ at a certain point in time is given by $k_jx_1^{c_{j1}}x_2^{c_{j2}}\dots x_n^{c_{jn}}$, where $k_j$ is a positive constant, referred to as the \emph{rate constant} for the reaction.  From stoichiometry, then, we know that the rate of change of $x_i$ due to the effects of reaction~$j$ is $(d_{ji}-c_{ji})k_j \prod_{i=1}^{n}x_i^{c_{ji}}$.  To simplify the notation, let us write $\mathbf{c_j}=(c_{j1},c_{j2}\dots , c_{jn})$, $\mathbf{d_j}~=~(d_{j1},d_{j2}\dots , d_{jn})$, and $\mathbf{u^v}=\prod_{\alpha}u_{\alpha}^{v_{\alpha}}$ for arbitrary vectors $\mathbf{u}=(u_{\alpha})$ and $\mathbf{v}=(v_{\alpha})$.  Summing over all reactions yields the complete rate equation \[\dfrac{dx_i}{dt} = \displaystyle\sum_{j=1}^{m}(d_{ji}-c_{ji})k_j\mathbf{x^{c_j}}.\]

Note that a reaction will only contribute to this sum when $c_{ji}$ and $d_{ji}$ are different.  As a special case, if $S_i$ does not participate in reaction $j$, then $c_{ji} = d_{ji} = 0$, and reaction~$j$ does not affect the concentration of $S_i$, as we should expect.  Henceforth, we will make the assumption that in fact $c_{ji}$ and $d_{ji}$ cannot both be positive, i.e. no species appears on both sides of the same reaction.  To make the notation cleaner, we will write $e_{ji} = d_{ji}-c_{ji}$ and $\mathbf{e_j}=(e_{j1},e_{j2}\dots , e_{jn})$, remembering that either $d_{ji}$ or $c_{ji}$ must be zero.   In terms of chemistry, this means that we require the processes inside the container to be broken down into sufficiently simple mechanisms so that each reactant in a given reaction undergoes a chemical change.\footnote{For further discussion, including the issue of catalysis, see section~\ref{Cat}.}

The changes in the chemical species over time can be condensed into a single function as follows:

\begin{defn} The \emph{vector rate function} for a chemical system is the polynomial function~$\mathbf{F}(\mathbf{x})$ from $(\mathbb{R}^+)^n$ to $\mathbb{R}^n$ defined by $\mathbf{F}(\mathbf{x}) = \left(\dfrac{dx_1}{dt}, \dfrac{dx_2}{dt}, \dots, \dfrac{dx_n}{dt}\right)$. \end{defn}

From $\mathbf{F}$, we define the associated Jacobian matrix $\mathbf{J_F}=(J_{ik})$, where
\begin{eqnarray*}
J_{ik}=\dfrac{\partial\mathbf{F}_i}{\partial x_k} = \dfrac{\partial}{\partial x_k} \left(\dfrac{dx_i}{dt}\right) = \displaystyle\sum_{j=1}^{m}\dfrac{c_{jk}}{x_k}e_{ji}k_j\mathbf{x^{c_j}}.
\end{eqnarray*}
We shall see in section~\ref{Domain} that it is valid to assume that the $x_i$ remain positive for all~$t$, so it is not a problem for them to appear in the denominators.  Denote by $J_{ijk}$ the $j^{th}$ term in $J_{ik}$: \[J_{ijk} = \dfrac{c_{jk}}{x_k}e_{ji}k_j\mathbf{x^{c_j}}.\] \noindent This quantity represents the effect of $S_k$ on the concentration of $S_i$ via reaction $j$.

\section{Remarks on the definitions}\label{remarks}

The framework introduced in section~\ref{Notation} contains several assumptions with important mathematical implications.  Here we will briefly discuss the chemical and mathematical rationale for several of our premises.

\subsection{Uniformity and continuity}

We will assume that our reaction vessel is isothermal and spatially homogeneous, and that the reactions take place in the liquid phase at constant density.  These conditions ensure that the only time-dependent quantities in the system are the concentrations of the reacting species.  In practice, depending on what type of container we are working with, the conditions may be not be perfectly satisfied: for example, a living cell may be close to isothermal, but it is far from well-mixed, although we often do not know enough to justify more precise assumptions.  Our hope is that networks of interest take place within a sufficiently uniform environment so that the qualitative behaviors of interest are unaffected by disregarding the heterogeneities.

A related assumption we will make is that the $x_i$ are continuous (in fact, differentiable) functions of time.  In real life, molecules exist in discrete units, and cells may have extremely small numbers of certain proteins, meaning that stochastic modeling can sometimes be more appropriate \cite{Bose}.  However, it is standard to assume that concentrations are continuous - especially in the lab - and this will be an essential condition for the results we present.  Once we assume continuity, differentiability follows naturally from the explicit rate functions.


\subsection{Mass-action kinetics}\label{massaction}

As stated in section~\ref{Notation}, we will be assuming that all of our reactions are written at the elementary level, where the law of mass action kinetics makes good physical sense.  Take, as an example, an elementary reaction $A + B \to C$.  One molecule of $A$ will combine with one $B$ to form a $C$ whenever the two reactant molecules collide with a certain energy and orientation.  If we assume that the probability of a collision is proportional to the abundances of the two species, then we immediately have the expression $k[A][B]$ for the reaction rate, where $k$ is a constant.  This is only an approximation, and in certain physical situations it will be more valid than others, but it is almost universally accepted \cite{MathPhys}.

Most of the time, especially in biology, an observer will see a composite process taking place rather than the mechanisms underlying it (see section~\ref{Cat} for a basic example), and it can be very difficult to write down all of the elementary reactions involved in a CRN.  The results in this paper will often be useful in discriminating among several proposed mechanisms for a process, based on the behaviors that each would permit according to the theory \cite{SCL}.

Mathematically speaking, the most important consequence of the law of mass action is that vector rate functions arising from CRNs will always be polynomials.  

\subsection{Catalysis and the positions of species in reactions}\label{Cat}

We assumed in section~\ref{Notation} that no species appears on both sides of the same reaction.  Our claim here is that any proposed reaction with a species as both a reactant and a product can be rearranged or decomposed into a simpler series of steps.

As an example, suppose we had a proposed reaction $A + B \to 2A + 2C$.  It could be that one molecule of $A$ is a spectator and does not participate chemically; in this case, the true reaction mechanism is $B \to A + 2C$.  It could also be that the molecule of $A$ on the reactant side undergoes a preliminary, implicit chemical change, so that the true mechanism consists of multiple steps, perhaps $A + B \to D + 2C$ followed by $D \to 2A$.

We note in particular that the conversion of a substrate $S$ into a product $P$ through the action of an enzyme or other catalyst $E$ can be decomposed into multiple reactions in this way.  The total effect of such reactions is $E+S \to E+P$, but the rates that are observed in catalysis are not of the form $k[E][S]$, so there must be some intermediate steps involved.  The most common model is the Michaelis-Menten mechanism, which involves a reversible enzyme-substrate-complex step: $E+S\leftrightharpoons ES \to E + P$.\footnote{In fact, the reaction $ES \to E + P$, like most chemical reactions, will probably be slightly reversible, but given the energetics of catalysis, it is almost always assumed to be unidirectional, especially when the product $P$ is being removed from the area of formation \cite{MathPhys}.}  Assuming that the concentration of $E$ is very small compared to that of $S$, the full reaction rate (i.e. for the formation of P) under this model is given by $k_1 [S]/(k_2+[S])$, where $k_1$ and $k_2$ are constants, and this expression matches observations quite well \cite{BH}.  Thus, while the process begins and ends with one molecule of $E$ present, in none of the elementary reactions does~$E$ or any other species appear on both sides.

Some of the best-known instances of multistability in biochemistry involve enzyme-catalyzed reactions, particularly when several enzymatic modules are coupled \cite{CF2, JG1}, as we will see, for example, in section~\ref{MAPK}.

\subsection{The species domain}\label{Domain}

It is clear that in the systems we are considering, the state variables, which represent concentrations of chemical species, cannot be negative.  In fact, the assumptions of continuity and mass-action kinetics allow us to restrict the domain to $(\mathbb{R}^+)^n$, as we stated in section~\ref{Notation}; this will be useful to us at several points in the discussion.\footnote{Craciun and Feinberg make this assumption as well, although without including a detailed analysis.}

If $x_i$ is initially zero for some $i$ (say at time $t = 0$) and $S_i$ is not produced via any of the reactions in our network, then $x_i$ will always be zero, and we can simply ignore~$S_i$ and all reactions in which it partcipates.  Otherwise, assume that~$x_i(0)>0$ and that~$x_i(t_0)=0$ for some $t_0>0$, where we are given some initial conditions such that the $x_i$ have finite solutions through time $t_0$.  We can choose $t_0$ such that $x_i(t_0-\epsilon) > 0$ for any $\epsilon > 0$.  As in section~\ref{Notation}, we have $dx_i/dt = \sum_{j=1}^{m}e_{ji}k_j\mathbf{x^{c_j}}$.  Since $S_i$ is assumed never to appear on both sides of the same reaction, we know that $e_{ji} < 0$ if and only if $\mathbf{x^{c_j}}$ contains a factor of~$x_i$, as these two conditions hold exactly when $S_i$ is on the reactant side of reaction~$j$.  So, we can write $dx_i/dt = -x_i \cdot P_1(t)+P_2(t)$, where $P_1$ and $P_2$ are polynomials with positive coefficients in terms of the $x_l$ (and thus implicit nonnegative functions of $t$) and $P_2$ does not contain any factors of $x_i$. If $0 < \epsilon < t' < t_0$, then dividing by $x_i$ and integrating yields, up to a constant, \[\ln x_i(t') = \displaystyle \int_{\epsilon}^{t'} \dfrac{dx_i}{x_i} = \displaystyle \int_{\epsilon}^{t'}\left(-P_1(t)+\dfrac{P_2(t)}{x_i}\right)dt.\]

As $t' \to t_0$, we have $x_i \to 0^+$, and so $\ln x_i \to -\infty$.  This means that $P_1(t') \to \infty$, in order to make the right-hand side approach $-\infty$, and hence the concentration of some species goes to infinity in finite time.  This is a contradiction, so we conclude that the concentration of $x_i$ must remain strictly positive for as long as the system has solutions.   



\section{Derivatives and conditions for injectivity}\label{DS}

In chemistry or thermodynamics, or even in everyday usage, an equilibium is a state of balance among the constituents of a system.  We will be interested in studying the equilibria (also referred to as the steady states) of a reaction network, specifically to determine whether or not there may be more than one.  To be precise, an \emph{equilibrium point} is a concentration vector $\mathbf{x}$ such that $\mathbf{F}(\mathbf{x}) = \mathbf{0}$.\footnote{The notion of steady states is sometimes taken to include limit cycles, but we will restrict the definition here to points.}  
Equivalently, if a CRN starts at an equilibrium point and evolves in time, it will not move from this point, as any change in a concentration due to some reaction is balanced by the other reactions in the network.

We will say that a network and its rate function are \emph{injective} if $\mathbf{F}(\mathbf{x_1}) \neq \mathbf{F}(\mathbf{x_2})$ whenever $\mathbf{x_1} \neq \mathbf{x_2}$.  This is the standard definition of an injective function.  Note that an injective CRN cannot have mulitiple equilibria.

For a two-dimensional dynamical system $dx/dt = f(x,y)$ and $dy/dt~=~g(x,y)$, the easiest way to check for equilibria is usually to graph the \emph{nullclines} $f(x,y) = 0$ and $g(x,y) = 0$ and to find their intersection points.  The stability properties of these steady states will also be fairly easy to determine.\footnote{Stability is a very rich and important problem and one that is obviously relevant to biology, but as such it deserves much more space than we could devote to it here.}  In real life, however, and especially in biochemistry, most CRNs of interest will have more than a few species.  With polynomial rate functions, there might be some hope of using tools from algebraic geometry, but in any case, the set of equations will be intimidating.  In the next few sections of this paper, we will begin to consider the question of how to determine, with minimal effort, when a given network is injective, with derivatives as our primary tool.

\subsection{Simple examples}

In one dimension, if $f: \mathbb{R} \to  \mathbb{R}$ is a differentiable function and $f(a) = f(b)$ with $a < b$, then Rolle's theorem (or the mean value theorem) tells us that $f'(c) = 0$ for some $a < c < b$.  Thus, if $f'(x)$ is nowhere zero, then $f(x)$ is injective.  This is about as simple as the relationship between injectivity and zero-valued derivatives can be; in higher dimensions, the story becomes much more complicated.  


Suppose $F$ is a differentiable, vector-valued function from $\mathbb{R}^n$ to  $\mathbb{R}^n$.  The \textquotedblleft derivative" of $F$ is now a matrix, the Jacobian $\mathbf{J}_F=(J_{ik})$, where $J_{ik}=\partial F_i/\partial x_k$.  By the inverse function theorem, $F$ is locally injective in the neighborhood of points where $\det(\mathbf{J}_F)$ is nonzero, but what can we say about its global properties?  In \cite{GN}, Gale and Nikaid\^{o} provide the following example: let $F:~\mathbb{R}^2 \to  \mathbb{R}^2$ be defined by $F(x, y) = (G(x, y), H(x, y))$, where $G(x, y) = e^{2x}-y^2+3$ and $H(x, y)~=~4e^{2x}y-y^3$.  Then,

\begin{eqnarray*}
\mathbf{J}_F
&=& \begin{pmatrix}
2e^{2x} & -2y  \\
8e^{2x}y &  4e^{2x}-3y^2
\end{pmatrix}\\
\Rightarrow \det(\mathbf{J}_F) &=& 2e^{2x}(4e^{2x}+5y^2) > 0\quad \forall \;(x,y) \in \mathbb{R}^2.
\end{eqnarray*}

However, $F(0,2)=F(0,-2)=(0,0)$.  Thus, even in two dimensions, an everywhere nonzero Jacobian determinant does not imply global injectivity.  In fact, this function was given in 1965 as a counterexample to a weaker conjecture by Samuelson \cite{Sam}, which stated that $F$ is injective if all of the leading principal minors of $\mathbf{J}_F$ - i.e., the determinants of the upper-left $i \times i$ submatrices - are nonvanishing.  So, any theorem of this sort will need to be weaker still.


\subsection{The Gale-Nikaid\^{o} theorem}

Having found a counterexample to Samuelson's conjecture, Gale and Nikaid\^{o} were able to formulate and prove two stronger conditions for global injectivity based on the Jacobian~$\mathbf{J}_F$~\cite{GN}.  Their work has subsequently inspired a number of extensions \cite{Parth}, but the original Gale-Nikaid\^{o} theorem is all that will interest us here.

We will say that $\mathbf{M}$ is a \emph{P-matrix} if \emph{all} of its principal minors (the determinants of the submatrices consisting of all entries with both indices in some given subset of $\{1, 2, \dots, n\}$) are strictly positive and that it is a \emph{weak P-matrix} if the determinant of $M$ is positive and all other principal minors are nonnegative.  

\begin{thm}[Gale-Nikaid\^{o} \cite{GN}]\label{GNThm1}
Let $D$ be a rectangular region of $\mathbb{R}^n$, and let $F: D \to  \mathbb{R}^n$ be a differentiable function.  If $\mathbf{J}_F$ is a P-matrix at all points $\mathbf{x} \in D$, then F is injective.  The same result holds if $\mathbf{J}_F$ is a weak P-matrix for all $\mathbf{x} \in D$, provided that $D$ is open.
\end{thm}

The proof for P-matrices is based on linear inequalities.  In particular, if we write $\mathbf{y} \geqq \mathbf{z}$ to mean that $y_i \ge z_i$ for each component $i$, it is shown first that the only solution to the system $\mathbf{J}_F\mathbf{y} \leqq \mathbf{0}, \; \mathbf{y} \geqq \mathbf{0}$ is $\mathbf{y} = \mathbf{0}$.  Then, more generally, for any $\mathbf{z} \in D$, the system $F(\mathbf{y}) \leqq F(\mathbf{z}), \; \mathbf{y} \geqq \mathbf{z}$ is only satisfied (in $D$) by $\mathbf{y} = \mathbf{z}$.

The proof for weak P-matrices, meanwhile, involves several ingredients from topological degree theory (also known as index theory).  First, if $\mathbf{z}$ is some point in $D$, then the function $G(\mathbf{y}) = F(\mathbf{z})+\mathbf{y}-\mathbf{z}$ clearly only takes the value $F(\mathbf{z})$ at $\mathbf{y} = \mathbf{z}$.  Also, the Jacobian of the function $\bar{F}(\mathbf{y}) = F(\mathbf{y}) + \lambda\mathbf{y}$ for $\lambda > 0$ is a P-matrix, so $\bar{F}$ is injective.  This allows us to construct a homotopy from $F$ to $G$ via $\bar{F}$, and hence the degrees of boundary cycles of $D$ are the same under $F$ and $G$.  The last useful fact is that these degrees equal the sums of the degrees of all points $\mathbf{y} \in D$ with $F(\mathbf{y}) = F(\mathbf{z})$ (respectively for $G$).\footnote{As observed in \cite{Text}, this theorem, which is cited as the \textquotedblleft Kronecker theorem on indices" in \cite{GN}, is remarkably similar to Gauss's law in electrostatics.}  Since $\det(\mathbf{J}_F)$ is positive and $\det(\mathbf{J}_G)$ clearly is as well, all of these degrees are~1.  Thus, because $G$ is injective on $D$, so is $F$.

\subsection{Polynomial maps and the real Jacobian conjecture}

The results we have seen so far are very general, applying to all differentiable functions.  The vector rate function for a CRN, however, has a specific form.  First and foremost, if we assume mass-action kinetics, then each component of the rate function is a polynomial in terms of the concentrations of the $n$ chemical species.  Because polynomials tend to be particularly well behaved, we might hope for a more powerful injectivity theorem in this special case.\footnote{For example, Campbell \cite{Campbell} has shown that Samuelson's conjecture holds for polynomial functions $F: \mathbb{R}^n \to \mathbb{R}^n$.}   Indeed, a great deal of effort has been devoted to the \textquotedblleft Jacobian conjecture" (for complex functions) and the related \textquotedblleft real Jacobian conjecture," which posits that a polynomial function $F: \mathbb{R}^n \to \mathbb{R}^n$ with nonvanishing Jacobian determinant is globally injective.  In 1994, however, Pinchuk proved that we are not so lucky, providing a counterexample, a polynomial function $F = (p,q): \mathbb{R}^2 \to \mathbb{R}^2$ with $\deg(p) = 10$ and $\deg(q) = 35$ that has a nowhere zero Jacobian determinant but is not injective \cite{Pinchuk}.

With the real Jacobian conjecture shown to be false, others have attempted to prove a modified version (see \cite{VDE} for a thorough reference).  One easy case is when all of the components of $F$ are first-degree polynomials: then $F$ is a linear transformation, and it is well known that $F$ is injective if and only if $\mathbf{J}_F$ is nonsingular.  If $F$ is quadratic, the conjecture still holds \cite{Wang}, and the proof is quite simple \cite{Oda}.  Beyond this, however, the theory becomes much harder; an example is the following, by Jelonek \cite{Jelonek}.  For a continuous map $F: X \to Y$, we say that $F$ is \emph{proper} at a point $y \in Y$ if there exists a neighborhood $U$ of $y$ such that $F^{-1}(\bar{U})$ is compact.  Let $S_F$ be the set of points in $Y$ at which $F$ is not proper.  Jelonek's theorem states that if $F: \mathbb{R}^n \to  \mathbb{R}^n$ is a polynomial map with nowhere zero Jacobian determinant and $S_F$ has codimension at least 3, then $F$ is injective.

This theorem and others like it, while providing a variety of potentially useful results on injectivity, tend to include conditions that are very difficult to apply to CRNs, for example in relation to the properties of the set $S_F$.  It is also significant that they are based on the requirement that $\det(\mathbf{J}_F) \ne 0$ for all $\mathbf{x} \in \mathbb{R}^n$ (or $\mathbb{C}^n$), whereas our domain of interest is only $(\mathbb{R}^+)^n$.

\section{Graphical conditions and Thomas-Soul\'{e}}\label{graphsection}

So far, we have seen the important role that the Jacobian of the vector rate function plays in determining the number of equilibria of a CRN.  In practice, though, the Jacobian determinant will be an immensely complicated polynomial, and finding its zeros will be a daunting task.  Might there be a way to use some simpler properties of the chemical system to determine, with limited effort, whether or not the Jacobian determinant can vanish in a certain domain and with certain rate constants?  Naturally, we expect any such conditions to sacrifice some generality for the sake of elegance.  In the sections that follow, we examine some approaches based on a network's interaction graph.

\subsection{The interaction graph and its cycles}

\begin{defn} The \emph{interaction graph} $G(\mathbf{x})$ of a CRN is the oriented graph with one vertex for each species $S_i$ and an edge from $S_{i_1}$ to $S_{i_2}$ ($i_1 \ne i_2$) if and only if $J_{i_2i_1}$ is nonzero.  Each edge is given a sign, which is the sign of $J_{i_2i_1}$. \end{defn}

Note that the signs of the edges in the interaction graph can change for different concentrations $\mathbf{x}$.  We will refer to an edge for which this occurs as \emph{concentration-ambiguous}.  A concentration-ambiguous edge from $S_{i_1}$ to $S_{i_2}$ will disappear when  $J_{i_2i_1} = 0$, but because our results are stated in terms of the possible configurations of $G(\mathbf{x})$ for any~$\mathbf{x}$, this will not matter.


Although the diagonal terms $J_{ii}$ of the Jacobian are usually nonzero, we omit self-edges from the graph because they are not used to form cycles (Definition~\ref{cycledef}).  This is another way of saying that autocatalysis and autoinhibition are expressed indirectly at the mass-action level (as in section~\ref{Cat}).  As we will see, diagonal terms are always nonpositive (and, under certain assuptions, always negative).

\begin{defn}\label{cycledef} A \emph{cycle} in an interaction graph is an ordered subset $(i_1, i_2, \dots, i_N)$ of $(1, 2, \dots, n)$ with $1 < N \le n$ such that there exists an edge from $S_{i_1}$ to $S_{i_2}$, from $S_{i_2}$ to $S_{i_3}$, and so on, including fron $S_{i_N}$ to $S_{i_1}$.  A cycle is classified as \emph{positive} or \emph{negative} based on the product of the signs of all of its edges. \end{defn}

In some examples, when the species names are more prominent than the indices, we will also use the notation $S_{i_1} \Rightarrow S_{i_2} \Rightarrow \dots \Rightarrow S_{i_1}$ to indicate a cycle.  Two cycles are \emph{disjoint} if they have no vertices in common; otherwise, they \emph{intersect}.   

We will be interested not only in the sign of a cycle but also in its \emph{sub-sign}.  For an edge from $S_{i_1}$ to $S_{i_2}$, we define its sub-sign, for some chosen reaction $j$ that contributes a term $J_{i_2ji_1}$ to $J_{i_2i_1}$, to be the sign of $J_{i_2ji_1}$.  The sub-sign of a cycle is then the product of the sub-signs of its edges, given $N$ chosen reactions: sub-positive if there are an even number of sub-negative edges and sub-negative if there are an odd number of sub-negative edges.  The sub-sign of an edge could depend on the choice of $j$, so we will always need to specify which reaction indices we are using.  If the edge is sub-positive for some $j$ and sub-negative for some other $j$, we refer to it as \emph{reaction-ambiguous}.  An edge could be both concentration-ambiguous and reaction-ambiguous (see in particular Lemma~\ref{AmbigLemma}).   However, because all of the concentrations $x_i$ are positive, the sign of $J_{i_2ji_1}$ does not depend on $\mathbf{x}$.

Given two cycles and choices of reactions $j$ as above for each of their edges, we will say that the cycles \emph{strongly intersect} if they have at least one vertex $S_i$ in common such that the reaction $j$ that is chosen for the outgoing edge from $S_i$ is the same for both cycles.  If this condition is not met, then they are \emph{weakly disjoint}.

\subsection{An example}\label{example}

To illustrate our definitions, consider the following hypothetical network, to which we will return in section~\ref{degenerate}:

\[\begin{array}{l}
1.\quad A + B + C \to X \\
2.\quad A + B + D \to Y \\
3.\quad C + E \to A \\
4.\quad D + E \to B \\
5.\quad A \to Z \\
6.\quad Z \to D.
\end{array}
\]

\noindent Using the law of mass action, we can compute the rate functions for each species:

\begin{eqnarray*}
\dfrac{d[A]}{dt} &=&  -k_1[A][B][C]-k_2[A][B][D]+k_3[C][E]-k_5[A]\\
\dfrac{d[B]}{dt} &=&  -k_1[A][B][C]-k_2[A][B][D]+k_4[D][E]\\
\dfrac{d[C]}{dt} &=&  -k_1[A][B][C]-k_3[C][E]\\
\dfrac{d[D]}{dt} &=&  -k_2[A][B][D]-k_4[D][E]+k_6[Z]\\
\dfrac{d[E]}{dt} &=&  -k_3[C][E]-k_4[D][E]\\
\dfrac{d[X]}{dt} &=&  k_1[A][B][C]\\
\dfrac{d[Y]}{dt} &=&  k_2[A][B][D]\\
\dfrac{d[Z]}{dt} &=&  k_5[A]-k_6[Z].
\end{eqnarray*}

\noindent The interaction graph for this network contains many cycles; one of interest is $C \Rightarrow A \Rightarrow Z \Rightarrow D \Rightarrow B \Rightarrow C$, which we will call cycle $c$.\footnote{An example of a cycle intersecting $c$ is $c': D \Rightarrow B \Rightarrow D$. If we pick the term $J_{224}$ in $c$ and the terms $J_{244}$ and $J_{422}$ in $c'$, then this is not a strong intersection.}   Numbering the species in alphebetical order from 1 to 8, the entries in the Jacobian which comprise this cycle are $J_{13}, J_{81}, J_{48}, J_{24},$ and $J_{32}$, which have the values

\begin{eqnarray*}
J_{13} \quad = \quad \dfrac{\partial}{\partial [C]}\dfrac{d[A]}{dt} &=&  -k_1[A][B]+k_3[E]\\
J_{81} \quad = \quad \dfrac{\partial}{\partial [A]}\dfrac{d[Z]}{dt} &=&  k_5\\
J_{48} \quad = \quad \dfrac{\partial}{\partial [Z]}\dfrac{d[D]}{dt} &=&  k_6\\
J_{24} \quad = \quad \dfrac{\partial}{\partial [D]}\dfrac{d[B]}{dt} &=&  -k_2[A][B]+k_4[E]\\
J_{32} \quad = \quad \dfrac{\partial}{\partial [B]}\dfrac{d[C]}{dt} &=&  -k_1[A][C].
\end{eqnarray*}

Let us suppose for later use that $k_1 = k_2$ and $k_3 = k_4$, remembering also that all of the $k_j$ are positive.  Using the terms $J_{133}, J_{851}, J_{468}, J_{224},$ and $J_{312}$, $c$ is sub-positive; the first three terms ($k_3[E], k_5$, and $k_6$) are positive, while the last two terms ($-k_2[A][B]$ and~$-k_1[A][C]$) are negative.  However, $J_{13}$ and $J_{24}$ are identical, so their signs are always equal, and hence their product is nonnegative and $c$ is always negative.

\subsection{The Thomas-Soul\'{e} theorem}

In 1981, Ren\'{e} Thomas conjectured that any dynamical system displaying stable oscillations must have at least one negative cycle, while any system with multiple steady states must contain a positive cycle \cite{Thomas1}.  Intuitively, these associations of behaviors with signed cycles make sense.  For example, suppose we have a three-species cycle $A~\Rightarrow~B~\Rightarrow~C~\Rightarrow~A$.  If all three edges are negative, then the cycle is negative.  In this case, if the concentration of $A$ is increasing, it will cause the concentration of~$B$ to decrease, which will in turn cause the concentration of $C$ to increase, which will cause the concentration of $A$ to decrease.  Since we initially assumed $A$ to be trending upward, we can see how the cycle promotes oscillatory behavior.  By contrast, if one edge is positive and the other two are negative, then the cycle is positive, and an increase in $A$ is self-reinforcing.  So, if the concentration of $A$ is perturbed away from equilibrium, the positive cycle will aid in pushing the system into a new basin of attraction.

In the years after Thomas's paper, a number of authors found the conjecture to be true in special cases \cite{Cinq1, Gouze, PMO, Snoussi}, and in 2003, Christophe Soul\'{e} presented a full proof~\cite{Soule}.  While the Thomas-Soul\'{e} theorem is certainly elegant, we should note that it gives no information about the sufficiency of positive cycles: in order for a system to display multistability, certain other conditions will need to be met, either in the structure of the network or in the ranges of values of specific state variables or parameters.

The precise statement of Soul\'{e}'s result is as follows.

\begin{thm}[Soul\'{e} \cite{Soule}]\label{SouleThm}
Let $D$ be an open, rectangular region of $\mathbb{R}^n$, and let $F: D \to  \mathbb{R}^n$ be a differentiable function.  If the interaction graph of F has no positive cycles for any $\mathbf{x} \in D$, then F is injective.
\end{thm}

The general idea behind the proof is that each term in a given principal minor with $N$ elements corresponds to a permutation $\sigma$ of $\{1, 2, \dots, N\}$ such that the $(i, \sigma(i))^{th}$ entry in the submatrix is nonzero.  This $\sigma$ in turn corresponds, by taking its algebraic cycle decomposition, to a union of disjoint cycles in our sense of the word (some possibly of length 1) that together incorporate all $N$ elements.  As we will see in section~\ref{CFsubsection1}, if all of the cycles in the interaction graph are negative, then all of the terms in the expansion of any principal minor will have the same sign, namely $(-1)^N$.  Neither the link between cycles and determinants \cite{Eco} nor this result about negative cycles \cite{BMQ} are new in Soul\'{e}'s work, but these facts are typically used in connection with the full Jacobian.  Because nonsingularity does not by itself imply injectivity, however, it is important that Soul\'{e} takes the results a step farther by involving the principal minors of $\mathbf{J}_F$.  

If we apply the above reasoning about cycles and determinants to the matrix $-\mathbf{J}_F$, we see that if $G(\mathbf{x})$ never contains a positive cycle, then none of the minors of this matrix will ever be negative.  If the full Jacobian determinant is strictly positive, then $-\mathbf{J}_F$ is a weak P-matrix, and by Gale-Nikaid\^{o} (Theorem~\ref{GNThm1}), $-F$ (and hence $F$) is injective.  If we assume homogeneous outflow (see section~\ref{CFsection}), then all minors have a nonzero diagonal term, and $-\mathbf{J}_F$ is in fact a P-matrix, meaning that the domain $D$ need not be open.  If neither of these assumptions is met, then the theorem still holds, with the caveat that it only guarantees a single \emph{nondegenerate} zero, i.e. one at which $\det(\mathbf{J}_F) \ne 0$ \cite{Soule, Soule2}.

Note that this theorem, like that of Gale-Nikaid\^{o}, applies to any differentiable function.  As a special case, of course, it can be applied to vector rate functions for reaction networks.  However, some of its power is lost when we confine our attention to this restricted class of polynomials.  The next theorem we will examine, by contrast, is formulated specifically for reaction networks.

\section{The Craciun-Feinberg theorem}\label{CFsection}

Recently, writing from the perspective of chemical engineering, Gheorghe Craciun and Martin Feinberg have presented an exciting, graph-based theorem on multistability in CRNs \cite{CF1}-\cite{CF3}.  They imagine their reactions to be taking place in what is called a continuous-flow stirred tank reactor (CFSTR), an isothermal, spatially uniform container with constant liquid flow streams in and out.  If $r$ is the rate of flow in units of 1/time, then the inward stream, or feed, is a constant vector $r\mathbf{f}$ (written with one component per species), with $\mathbf{f}$ in units of moles per volume, while the outward stream, or outflow, is assumed to be drawn homogeneously from the reactor and hence equals $r\mathbf{x}$.  We can think of the feed and the outflow as reactions of the form $0 \to S$ and $S \to 0$, respectively, to be appended to the CRN, with the outflow, we note, obeying mass action kinetics.  For example, in the network from section~\ref{example}, the rate function for $A$ would become

\begin{eqnarray*}
\dfrac{d[A]}{dt} &=&  rf_1-k_1[A][B][C]-k_2[A][B][D]+k_3[C][E]-k_5[A]-r[A].\\
\end{eqnarray*}

\noindent A network with this type of outflow is typically referred to as \emph{open} or \emph{homogeneous}.

Upon computing any partial derivatives of the rate expressions, the feed term disappears, so the feed does not affect the Jacobian.  From now on, we will ignore the feed stream, noting that because it is constant, the feed-free reaction network is injective if and only if the full CFSTR network is injective \cite{CF1}. On the other hand, the outflow contributes to the diagonal terms of the Jacobian matrix, and this fact is used in the proof.  One might argue that the requirement of outflow reactions makes the theory less widely applicable, but a few points may be offered in response.  First, in a biological context, the outflow may be interpreted as the degradation of proteins or RNAs, which, as with the flow from a CFSTR, is typically assumed to proceed at a rate proportional to the concentrations of the species in question \cite{switch, Immune, Soule2}.\footnote{As Craciun and Feinberg observe, their results are unchanged when the rate constants for the outflow are different for the different species \cite{CF1}.}  Additionally, even if we wish to insist that the concentrations of certain enzymes or other species are not subject to outflow or degradation, Craciun and Feinberg prove that their results carry over, with only degenerate exceptions, to this \textquotedblleft entrapped species" case \cite{CF3}.

In addition to the presence of outflow, the Craciun-Feinberg theory differs from other approaches in the way that it treats rate constants and in how it parses the Jacobian matrix.  So far, all of our results have been formulated to apply to a single, well-defined network, complete with rate constants.  In Craciun-Feinberg, however, when one is able to prove injectivity for a CRN, it applies for any positive values of the $k_j$.  This added power can be useful, especially in biology, because accurate values for rate constants and concentrations can be very difficult to obtain, and realistic reaction networks will usually contain many such parameters \cite{Bailey, Soule}.  On the other hand, though, we might want to be able to distinguish, for a given network, which parameter regimes can support multistability and which cannot.  The structure of the proof in section~\ref{CFsubsection2} is very much related to the parameter-independent nature of the theory.


With regard to the Jacobian, while Craciun and Feinberg use different terminology from ours, their key innovation is effectively to split each entry $J_{ik}$ into its component terms $J_{ijk}$ and to expand $\det(\mathbf{J_F})$ in terms of these components.  This allows for finer criteria in relating the structure of $G(\mathbf{x})$ to the value of $\det(\mathbf{J_F})$; some (sub-)positive cycles can be accounted for without violating the conclusion of injectivity.  Section~\ref{CFsubsection1} is devoted to describing how this process works.

\subsection{New notation and the statment of the theorem}\label{CFnotation}

Craciun and Feinberg's statements and proofs are somewhat different from ours, both in choice of language and in some more significant structural ways \cite{CF1, CF2}.  As a result, we will begin by explaining a few of their concepts, in order to illuminate the relationships between their work and the more standard interaction-graph-based literature (about which we will say more in section~\ref{compare}).

First, Craciun and Feinberg define two new graphs, called the SR and OSR Graphs, which have vertices representing both reactions and species as well as more elaborate edge labels, and they make use of the cycles in these graphs (SR in the statement of the theorem and OSR in the proof) rather than those in the standard interaction graph.  This discrepancy is the reason for our introduction of sub-signs in the interaction graph.  Second, they classify cycles on the basis of parity, rather than sign, as follows.  Let $(i_1, i_2, \dots, i_N)$ be a cycle in $G(\mathbf{x})$, and let $j_1, j_2, \dots, j_N$ be reaction indices such that~$J_{i_{h+1}j_hi_h}$ is nonzero for $1 \le h \le N$ (cyclically).  Each index $j_h$ is called \emph{shared} if $S_{i_h}$ and $S_{i_{h+1}}$ are on the same side of reaction $j_h$.  Then, we say that $(i_1, i_2, \dots, i_N)$ has \emph{even potential} if, for some choice of the indices $j_1, j_2, \dots, j_N$, an even number of them are shared.  For example, in our use of the cycle $c$ in section~\ref{example}, the shared edges are those from $D$ to $B$ and from $B$ to $C$, so~$c$ has even potential. 

Having given a sense of their framework, we will now show how their definitions relate to the notions we have already introduced.  Let $(i_1, i_2, \dots, i_N)$ and $j_1, j_2, \dots, j_N$ be as above, so that, for each $h$, \[J_{i_{h+1}j_hi_h} = \dfrac{c_{j_hi_h}}{x_{i_h}}e_{j_hi_{h+1}}k_{j_h}\mathbf{x^{c_{j_h}}}\] is nonzero.  This implies that $c_{j_hi_h}$ is nonzero, which is to say that $S_{i_h}$ is on the reactant side of reaction $j_h$.  The sign of $J_{i_{h+1}j_hi_h}$ is determined by that of $e_{j_hi_{h+1}}$, so this term is positive when $S_{i_{h+1}}$ is on the product side of reaction $j_h$ and negative when $S_{i_{h+1}}$ is on the reactant side.  (Recall that a species cannot appear on both sides of a reaction.)  Thus, the edge $J_{i_{h+1}i_h}$ is sub-negative if and only if $j_h$ is shared.  By counting the number of such edges, we see that $(i_1, i_2, \dots, i_N)$ has even potential if and only if there exist $j_1, j_2, \dots, j_N$ such that $(i_1, i_2, \dots, i_N)$ is sub-positive.

As one last step, we define the \emph{value} of the cycle $(i_1, i_2, \dots, i_N)$, given a choice of reaction indices $(j_1, \dots, j_N)$, to be the quotient

\[\left| \dfrac{\displaystyle\prod_{h=1}^NJ_{i_{h+1}j_hi_h}}{\displaystyle\prod_{h=1}^NJ_{i_hj_hi_h}} \right| = \left| \dfrac{\displaystyle\prod_{h=1}^Ne_{j_hi_{h+1}}}{\displaystyle\prod_{h=1}^Ne_{j_hi_h}} \right|.\]  

\noindent This expression amounts to an alternating product and quotient of the stoichiometric coefficients of the species in the chosen reactions.  Its significance will become apparent in the next section.


We are now able to state the main theorem.

\begin{thm}[Craciun-Feinberg \cite{CF2}, restated]\label{CFlemma1}
Suppose an open (homogeneous) CRN has the property that if a cycle in $G(\mathbf{x})$ is sub-positive for some chosen reaction indices, then it has value 1.  If no two sub-positive cycles strongly intersect, then the network cannot have multiple equilibria, regardless of the rate constants of the reactions.
\end{thm}

In particular, if no cycles have even potential, then the network is injective.  The same is true if all stoichiometric coefficients are equal to 1 and no species appears in more than two reactions~\cite{CF2}.

The proof can be broken into two components: showing that the cycle condition implies a nonzero Jacobian determinant and showing that the nonzero Jacobian determinant implies injectivity.  

\subsection{The Craciun-Feinberg proof from cycles to the Jacobian}\label{CFsubsection1}

\begin{thm}\label{CFlemma2}
Suppose an open (homogeneous) CRN has the property that if a cycle in~$G(\mathbf{x})$ is sub-positive for some chosen reaction indices, then it has value \emph{less than or equal} to 1.  If no two sub-positive cycles strongly intersect, then the Jacobian $\det(\mathbf{J_F})$ is nonzero for all $\mathbf{x}$ and all positive values of the rate constants.
\end{thm}

Note that we have generalized Craciun and Feinberg's result to include a wider range of sub-positive cycles.  Na\"{i}vely, we would expect that the new version now covers, in addition to the cycles with value 1, roughly half of all other cycles.  

\begin{proof}

In the natural expansion of $\det(\mathbf{J_F})$, each term is a product of $n$ of the entries~$J_{ik}$, with one term for each permutaion $i_1, i_2, \dots, i_n$ of $1, 2, \dots, n$.  Craciun and Feinberg, by essentially splitting the $J_{ik}$ into sums of $J_{ijk}$, cleverly rearrange the expansion into terms corresponding to a permutation $i_1, i_2, \dots, i_n$ \emph{and} some list $j_1, j_2, \dots, j_n$.\footnote{Because of the outflow, there are always at least $n$ reactions in the network.}  Each of these new terms turns out to be a product of some positive rate constants and concentrations, multiplied by $\det([\mathbf{c_{j_1}}, \dots, \mathbf{c_{j_n}}])\cdot \det([\mathbf{e_{j_1}}, \dots, \mathbf{e_{j_n}}])$, with these matrices of $\mathbf{e}$s formed by stacking the $n$ row vectors listed.\footnote{In Craciun-Feinberg, these matrices are given as the transposes of the versions here, but we can make the change because taking transposes preserves determinants and flips the directions of all edges, so that cycles and their intersections are preserved, with opposite orientation.}   We would like to show that all such coefficients - the products of the determinants - are zero or have sign $(-1)^n$.  



Let us consider a nonzero product $\det([\mathbf{c_{j_1}}, \dots, \mathbf{c_{j_n}}])\cdot \det([\mathbf{e_{j_1}}, \dots, \mathbf{e_{j_n}}])$.  In the first determinant, there must be a nonzero entry in each row and in each column, so we can switch the order of the $i_h$ such that the diagonal entries (of both matrices, in fact) are all nonzero.  Different orderings satisfying this condition correspond to nonzero products involving different permutations of the same $n$ reactions $j_1, j_2, \dots, j_n$.

Now consider $\det([\mathbf{e_{j_1}}, \dots, \mathbf{e_{j_n}}])$.  From our discussion of Theorem~\ref{SouleThm}, we know that the nonzero terms in its (natural) expansion correspond bijectively to all sets of disjoint cycles in the interaction graph for $[\mathbf{e_{j_1}}, \dots, \mathbf{e_{j_n}}]$.  For this matrix, a cycle consists of indices $i_1, i_2, \dots, i_N$ such that each $e_{j_hi_{h+1}}$ is nonzero, which is to say that $S_{i_{h+1}}$ appears in reaction~$j_h$.   Since the diagonal entries are nonzero, we know, as in section~\ref{CFnotation}, that~$S_{i_h}$ is on the reactant side of reaction $j_h$ for all $h$, and so $c_{j_hi_h}$ and $J_{i_{h+1}j_hi_h}$ are nonzero.  Thus, the cycles here are exactly the same as those in $G(\mathbf{x})$ such that, not only is $J_{i_{h+1}i_h}$ nonzero, but $J_{i_{h+1}j_hi_h}$ is nonzero as well.  Moreover, the signs of the cycles here are the same as the original sub-signs, since $e_{j_hi_{h+1}}$ has the same sign as $J_{i_{h+1}j_hi_h}$, as observed in section~\ref{CFnotation}.  Note that if any of these cycles, when regarded in $G(\mathbf{x})$, contain $S_{i_h}$, then the outgoing edge from $S_{i_h}$ must correspond to reaction $j_h$, meaning that any intersecting cycles from~$[\mathbf{e_{j_1}}, \dots, \mathbf{e_{j_n}}]$ are strongly intersecting in $G(\mathbf{x})$.

Suppose we examine a term $T$ in $\det([\mathbf{e_{j_1}},~\dots,~\mathbf{e_{j_n}}])$ having $c~=~(i_1, i_2, \dots, i_N)$ among its set of corresponding disjoint cycles.  This cycle $c$ contributes a factor of $C~=~(-1)^{N-1}\prod_{h=1}^Ne_{j_hi_{h+1}}$ to the term (where $C = T$ only if $N = n$), where the leading~$(-1)^{N-1}$ is due to the sign of the permutation associated with $c$.  By the correspondence of cycles with those of $G(\mathbf{x})$, we know that if $c$ is sub-negative, then the sign of~$C$ is $(-1)^{N-1}(-1) = (-1)^N$, and if $c$ is sub-positive, then the sign of~$C$ is $(-1)^{N-1}$.  There also exists another nonzero term $T'$ in the expansion which differs from $T$ by replacing $C$ with $C' = \prod_{h=1}^Ne_{j_hi_h}$, i.e. by taking the identity permutation of the set $\{i_1, i_2, \dots, i_N\}$.  Since $e_{j_hi_h} = -c_{j_hi_h}$ for all $h$, the sign of $C'$ must be $(-1)^N$.

If all of the cycles contributing to $T$ are sub-negative, then the sign of $T$ is $(-1)^n$.  More generally, suppose that for some choice of reactions $j_1, j_2, \dots, j_n$ as above, we have a set $S_P$ of sub-positive cycles, all of which are weakly disjoint from each other and have $|C'|>|C|$.  Let $S_N$ be any set of disjoint sub-negative cycles, with total associated product~$C_N$, and suppose there are $M$ cycles in $S_P$ that are disjoint from all elements of $S_N$.  If these cycles have associated factors $C_1, C_2, \dots, C_M$ and $C_1', C_2', \dots, C_M'$ as above, then the sum of all of the terms in the expansion of the determinant that contain exactly these sub-negative cycles and some subset of the $M$ sub-positive cycles is $C_N\prod_{k=1}^M(C_k+C_k')$, since each sub-positive cycle either is or is not present.  By the condition $|C'|>|C|$, this quantity has the same sign as $C_N\prod_{k=1}^MC_k'$, which is $(-1)^n$.  Taking all sets $S_N$, we thus account for all of the terms in the expansion of $\det([\mathbf{e_{j_1}},~\dots,~\mathbf{e_{j_n}}])$, and hence the sign of the determinant is $(-1)^n$.\footnote{Craciun and Feinberg \cite{CF2} state and prove this result with the requirement that $|C'|=|C|$, using strict cancellation of terms.}

In order to complete the proof of Theorem~\ref{CFlemma2}, we need to show that $\det([\mathbf{c_{j_1}}, \dots, \mathbf{c_{j_n}}])$ is positive.  Luckily, the same argument we have just used applies again here, but with two modifications.  First, the full set of possible cycles is reduced from all of those in $G(\mathbf{x})$ to those having all edges sub-negative.  Second, the entries in $[\mathbf{c_{j_1}}, \dots, \mathbf{c_{j_n}}]$ are all positive, meaning that the diagonal term in the determinant, and hence the entire determinant, has sign $+1$ instead of $(-1)^n$. 

\end{proof}

\subsection{The Craciun-Feinberg proof from the Jacobian to injectivity}\label{CFsubsection2}

Once we know that all of the terms in the determinant of $\mathbf{J_F}$ have the same sign, all we need in order for the determinant to be nonvanishing is that some term is nonzero.  For this, we simply take the product of the diagonal elements, all of which are strictly negative.  We would then like to prove that the nonzero Jacobian determinant implies injectivity.  The method here is taken directly from Craciun and Feinberg \cite{CF1}, but  whereas they only deal with CRNs, we will expand the proof to a more general class of functions.

Let $m_i$ and $v_{ij}$ be nonnegative integers for $i = 1, 2, \dots, n$ and $j = 1, 2, \dots, m_i$.  For real numbers $A = (\alpha_{ij})$, let $\mathbf{F}^A(\mathbf{x}): (\mathbb{R^+})^n \to \mathbb{R}^n$ be a polynomial function with components $\mathbf{F}^A_i(x_1, x_2, \dots, x_n) = \sum_{j=1}^{m_i}\alpha_{ij}\mathbf{x^{v_{ij}}}$, where we use the shorthand $\mathbf{u^v}=\prod_{\gamma}u_{\gamma}^{v_{\gamma}}$.

From $\mathbf{F}^A$, remembering that the $x_i$ are positive, we define the associated Jacobian matrix $\mathbf{J}_{\mathbf{F}^A}=(J_{ik})$, where \[J_{ik}=\dfrac{\partial\mathbf{F}^A_i}{\partial x_k} = \displaystyle\sum_{j=1}^{m_i}\dfrac{v_{ijk}}{x_k}\alpha_{ij}\mathbf{x^{v_{ij}}}.\] 

\begin{thm}\label{T}
Fix $m_i$ and $\mathbf{v_{ij}}$ as above. The following are equivalent. 1: The function $\mathbf{F}^A$ is injective for all $A$. 2: The determinant of $\mathbf{J}_{\mathbf{F}^A}$ is nonzero for all $A$ and all $\mathbf{x} \in (\mathbb{R^+})^n$.
\end{thm}

\begin{proof}

Suppose that for some values of $\mathbf{x}$ and the $\alpha_{ij}$, the determinant of $\mathbf{J}_{\mathbf{F}^A}$ is zero.  Since $\mathbf{J}_{\mathbf{F}^A}$ is a linear transformation, this means that there is some nonzero vector $\mathbf{y} \in \mathbb{R}^n$ such that $\mathbf{J}_{\mathbf{F}^A}\mathbf{y} = \mathbf{0}$.  From above, the vector $\mathbf{J}_{\mathbf{F}^A}\mathbf{y}$ has components \[(\mathbf{J}_{\mathbf{F}^A}\mathbf{y})_i = \displaystyle\sum_{k=1}^{n}\left(\displaystyle\sum_{j=1}^{m_i}\dfrac{v_{ijk}}{x_k}\alpha_{ij}\mathbf{x^{v_{ij}}}\right)y_k,\] which we assume are all zero.  We claim that these equalities are equivalent to \[\displaystyle\sum_{j=1}^{m_i}\displaystyle\sum_{k=1}^{n}\eta_{ij}v_{ijk}\delta_k = 0\] for all $i$, where $\eta_{ij} = \alpha_{ij}\mathbf{x^{v_{ij}}}$ and $\delta_k = y_k/x_k$.  All of these are real numbers, and not all of the $\delta_k$ are zero.  In the forward direction, the substitutions are straightforward.  In the reverse direction, we can obtain the implication by, for example, taking all of the $x_k$ to be 1 and then letting $\alpha_{ij} = \eta_{ij}$ and $y_k = \delta_k$.

Next, we claim that this last set of equalities is equivalent to the set \[\displaystyle\sum_{j=1}^{m_i}K_{ij}\cdot \left(\exp\left(\displaystyle\sum_{k=1}^{n}v_{ijk}\delta_k\right)-1\right) = 0\] for all $i$, where \[K_{ij} = \dfrac{\eta_{ij}\cdot\displaystyle\sum_{k=1}^{n}v_{ijk}\delta_k}{\exp\left(\displaystyle\sum_{k=1}^{n}v_{ijk}\delta_k\right)-1}.\] Again, the forward substitution is no problem, with the caveat that the denominators may be zero, in which case we just choose arbitrary $K_{ij}$ for those indices, and the equalities hold regardless of what we choose.  The reverse substitution, for the $\eta_{ij}$ in terms of the~$K_{ij}$, is analagous.

From this set of equalities, we can pass to \[\displaystyle\sum_{j=1}^{m_i}K_{ij}\cdot \left(\dfrac{\mathbf{a^{v_{ij}}}}{\mathbf{b^{v_{ij}}}}-1\right) = 0\] for all $i$, taking all the $a_k$ and $b_k$ positive and $a_k/b_k = e^{\delta_k}$ for each $k$.  In reverse, we simply let $\delta_k = \ln(a_k/b_k)$ for each $k$, noting that the condition that not all of the $\delta_k$ are zero is equivalent to the condition that $\mathbf{a} \neq \mathbf{b}$.  Finally, the above are equivalent to the relations \[\displaystyle\sum_{j=1}^{m_i}L_{ij}\cdot \left(\mathbf{a^{v_{ij}}}-\mathbf{b^{v_{ij}}}\right) = 0\] for all $i$, where $L_{ij}=K_{ij}/\mathbf{b^{v_{ij}}}$, or equivalently, $K_{ij}=L_{ij}\mathbf{b^{v_{ij}}}$.  And, at last, we see that our train of equivalences, starting at a determinant of zero for some point in the domain and some coefficients, ends with a lack of injectivity for some (different) values of the coefficients.  This proves the theorem.

\end{proof}

For the case of a CRN, with reactions of the form $c_{j1}S_1 + c_{j2}S_2 + \dots + c_{jn}S_n \to d_{j1}S_1 + d_{j2}S_2 + \dots + d_{jn}S_n$, we would have $\mathbf{F}_i = dx_i/dt = \sum_{j=1}^{m}e_{ji}k_j\mathbf{x^{c_j}}$, so that $m_i = m$ for all $i$, $\alpha_{ij} = e_{ji}k_j$, and $\mathbf{v_{ij}} = \mathbf{c_j}$ for each $i$.  Note that in the proof above, we can restrict our coefficients to be positive by using the fact that $r(e^r-1) \ge 0$ for all real numbers $r$, with equality only when $r = 0$.

\subsection{Applying the theorem}\label{inhibition}

As discussed at the beginning of section~\ref{CFsection}, the assumption of outflow for CRNs is critical to the Craciun-Feinberg theory (and well justified in many cases).  Experimental and theoretical evidence suggests, though, that closed systems, having at most partial outflow, are easier to constrain in their behavior than are open ones \cite{ClarkeShort, Leib, SCL}.  Feinberg's deficiency theory, for example, gives extensive information for networks with deficiency~0, partial information for deficiency 1, and no information for higher values, and closed systems tend to have low deficiency \cite{Bailey, Feinberg1, Feinberg2}.  Additionally, examples of multistability in the laboratory have tended to come from closed systems \cite{CF1}.

To illustrate the usefulness of the Craciun-Feinberg theorem, however, let us consider the first experimental example of bistability in an open homogeneous system (at least according to its publishing author), which was found by Degn~\cite{Degn} in the peroxidase-catalyzed oxidation of NADH.  He assumes a very simple mechanism, consisting only of Michaelis-Menten enzyme action with one additional substrate-inhibition step:\footnote{The irreversibility of the first step is irrelevant.}

\[\begin{array}{ll}
1. & E + S \to ES \\
2. & ES \to E + P \\
3, 4. & ES + S \leftrightharpoons ES_2.
\end{array}\]

\noindent The form $ES_2$ is a deactivated enzyme-substrate complex.  While we cannot prove that the network supports bistability, we can note that, despite its simplicity, the interaction graph contains two strongly intersecting sub-positive cycles: $E \Rightarrow ES \Rightarrow E$ (using reactions 1 and 2, both sub-positive edges) and $E \Rightarrow S \Rightarrow ES \Rightarrow E$ (using reactions 1, 3, 2, yielding sub-negative, sub-negative, and sub-positive edges). 

By contrast, consider two possible schemes involving product inhibition:

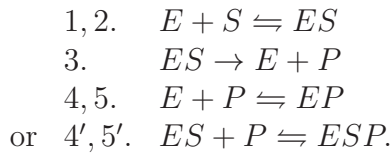
\[\begin{array}{cll}
&1, 2. & E + S \leftrightharpoons ES \quad \\
&3. & ES \to E + P \\
&4, 5. & E + P \leftrightharpoons EP \\
\text{or} \; &4', 5'. & ES + P \leftrightharpoons ESP.
\end{array}\]

\noindent Neither of these mechanisms contains any stoichiometric coefficients greater than 1, and for neither one does the interaction graph have two strongly intersecting nontrivial\footnote{See section~\ref{tricky}.} positive cycles.  Thus, by Theorem~\ref{CFlemma1}, the networks are injective.  As a result, we can confirm for Degn that the ability of the peroxidase system to generate bistability cannot be due to product inhibition, even though all three mechanisms discussed here have very similar forms. 
\medskip
\medskip

\section{Comparing and reconciling the approaches}\label{compare}

The theorems discussed above come from a wide range of fields.  Jelonek and Pinchuk approach the real Jacobian problem as pure mathematicians, Clarke writes as a chemical physicist, Gale and Nikaid\^{o} are mathematical economists, Thomas and Soul\'{e} come with a biological perspective, and Craciun and Feinberg are chemical engineers.  Given the similarities in the methods of the Gale-Nikaid\^{o}/Thomas-Soul\'{e} and the Craciun-Feinberg theorems in particular, though, it is not surprising that the results are closely linked.  

\subsection{Introduction to the case of no concentration-ambiguity}\label{ambig_intro}

The relationship between the Thomas-Soul\'{e} and Craciun-Feinberg theorems is particularly strong when all of the elements in the Jacobian have a fixed sign for all values of the state variables (or, in our terminology, no edges in the interaction graph are concentration-ambiguous). 
In general, the assumption of no concentration-ambiguity is a fairly restrictive one, since the entries in the Jacobian can be complicated functions.  As we will see, though, it has been made in numerous places, albeit without extensive justification.  Our analysis begins with the following lemma.

\begin{lemma}\label{AmbigLemma}
In the interaction graph of a CRN, an edge is concentration-ambiguous if and only if it is reaction-ambiguous.
\end{lemma}

\begin{proof}

First, suppose we have a CRN such that the edge from $S_{i_1}$ to $S_{i_2}$ is not reaction-ambiguous, which is to say that all of the terms $J_{i_2ji_1}$ have the same sign.  Since all species concentrations are positive, we see immediately that $J_{i_2i_1}$ has the same sign for all $\mathbf{x}$, and hence the edge is not concentration-ambiguous.

Next, suppose the edge from $S_{i_1}$ to $S_{i_2}$ is reaction-ambiguous; we would like to show that as $\mathbf{x}$ varies, $J_{i_2i_1}$ is sometimes positive and sometimes negative.  Given a nonzero term \[J_{i_2ji_1} = \dfrac{c_{ji_1}}{x_{i_1}}e_{ji_2}k_j\mathbf{x^{c_j}},\] \noindent its sign is determined by that of $e_{ji_2}$.  Looking at reaction $j$ (as in section~\ref{CFnotation}), we see that~$J_{i_2ji_1}$ is positive when $S_{i_2}$ is on the product side and negative when $S_{i_2}$ is on the reactant side.

Note that the rate term above contains a factor of $x_{i_2}$ if and only if $c_{ji_2} > 0$ in the exponent, i.e. if $S_{i_{2}}$ is on the reactant side of reaction $j$.  As a result, in the limit of very small values of $x_{i_2}$, the sum of the negative terms in $J_{i_2i_1}$ will approach 0, while in the limit of very large values, this sum will approach $-\infty$.  Meanwhile, the positive terms in $J_{i_2i_1}$ are unaffected as $x_{i_2}$ varies.  By continuity, then, the edge from $S_{i_1}$ to $S_{i_2}$ is concentration-ambiguous.

\end{proof}

This lemma gives us a clear characterization of concentration-ambiguity for CRNs, namely that it obtains precisely when a species $S_{i_1}$ serves as a direct activator for another species $S_{i_2}$ in one reaction and as an inhibitor of it in another reaction.  Based on this condition, we might indeed expect to find chemical systems without concentration-ambiguity reasonably often, although there are certainly well-known cases that are ambiguous, especially in gene regulation \cite{Cinq1}.  For example, to use we have seen, the first product-inhibition system discussed in section~\ref{inhibition} has no concentration-ambiguity, while the second does.  The first system also provides another illustration of the importance of assuming that enzyme-driven processes are broken down into elementary reaction mechanisms (as in section~\ref{Cat}), since the edge from $E$ to $P$ would be concentration-ambiguous if the first two reactions were replaced by the single step $E+S \to E+P$. 

\subsection{Thomas-Soul\'{e} and Craciun-Feinberg}\label{degenerate}

Before Soul\'{e}'s completion of Thomas's conjecture, Gouz\'{e} \cite{Gouze}, Plahte-Mestl-Ohmolt \cite{PMO}, and Snoussi \cite{Snoussi} all published calculus-based proofs of the conjecture with the additional assumption of no concentration-ambiguity.\footnote{The result of Cinquin-Demongeot \cite{Cinq1} is somewhat more general.}  The Craciun-Feinberg theorem also gives us a simple proof of the conjecture in this case, as follows.  Suppose that for a given CRN, $G(\mathbf{x})$ has no positive cycles and no concentration-ambiguous edges for any $\mathbf{x}$.  The sub-sign of any cycle will always be the same as its sign, i.e. it will be negative.  Thus, Theorem~\ref{CFlemma1} immediately tells us that, as long as our system is open, it must be injective.  In fact, we obtain the stronger result that injectivity holds regardless of the rate constants.  Our conclusion is most similar to that of Snoussi, since he also assumes (in effect) an outflow or degradation term for each state variable \cite{Snoussi}.


We can also make some (but not complete) progress with this argument when edges are allowed to be concentration-ambiguous.  Again, we would like to show that negative cycles must be sub-negative, so that we can apply Theorem~\ref{CFlemma1}.  Suppose $C = (i_1, i_2, \dots, i_N)$ is a negative cycle that is sub-positive for some choice of reactions $j_1, j_2, \dots, j_N$.  There must be at least one edge in $C$ that is reaction-ambiguous, or else $C$ would be positive as well as sub-positive.  Without loss of generality, let us say that $J_{i_2i_1}$ has some terms with positive signs and some with negative signs.


Let us vary $x_{i_2}$, as in the proof of Lemma~\ref{AmbigLemma}, until there is a change in the sign of $J_{i_2i_1}$.  Unless the sign of another edge in $C$ changes simultaneously, we will have a contradiction, because $C$ will be positive for slightly larger or smaller $x_{i_2}$.  In other words, at any concentration vector such that $J_{i_2i_1} = 0$, we must have $J_{i_{h+1}i_h} = 0$ for some other edge in the cycle.  This gives us a finite set of polynomials such that all of the zeroes of one of them ($J_{i_2i_1} = 0$) that are in $(\mathbf{R}^+)^n$ are zeroes of at least one of the others.  This situation is mathematically possible, but it seems physically unrealistic, a  \textquotedblleft measure-zero" case.  


It is true, however, that this case could arise.  For an example, we turn back to the CRN in section~\ref{example}, assuming outflow reactions for each species in addition to the six reactions shown.  As noted above, the cycle $C \Rightarrow A \Rightarrow Z \Rightarrow D \Rightarrow B \Rightarrow C$ is always negative, yet it is sub-positive for the correct choice of reaction indices.  This example is also a good illustration, however, of our statement that this is a \textquotedblleft degenerate" or measure-zero condition:  mathematically, there may be nothing exceptional about the system, but in real life, we could not expect any pairs of rate constants to be precisely equal.


\subsection{Monotone systems and P\'{o}lya's permanent problem}\label{monotone}

It is worth taking a moment to acknowledge two other sources of related results.  The first is the theory of monotone systems, as discussed in the work of Angeli and Sontag \cite{AS1, ASmulti, AFS}.  A CRN gives rise to a monotone system provided that none of the edges in its interaction graph are concentration-ambiguous and all of the cycles are \emph{positive}.  One would expect these conditions to rule out a rather large class of networks, but when they apply, the theory gives more extensive information about sufficient conditions for multistability and the locations of steady states than do any of our other approaches.  The example of MAPK cascades figures prominently in their papers; see section~\ref{MAPK} for a discussion.  A special subcategory of monotone systems is that of cooperative systems, which are those with interaction graphs having all edges positive for all $\mathbf{x}$ \cite{AS1}.

P\'{o}lya's permanent problem asks when it is possible to relate the determinant of a matrix to its \emph{permanent}, which is formed by adding the terms in the standard expansion of the determinant without their associated permutation-signs attached.  This question is not directly one of injectivity, but as in section~\ref{CFsubsection1}, it is very much related to the interplay between matrices with signed elements and their related graphs \cite{BMQ}.  In fact, the problem has appeared in a number of contexts and forms, even including the calculation of resonance structures in organic chemistry \cite{Orgo}. Combinatorial extensions of the problem have been treated extensively in the graph theory literature \cite{Mcc1, Mcc2}, especially with regard to conditions for the existence of various cycles in a number of different kinds of graphs.

\subsection{Craciun-Feinberg and Gale-Nikaid\^{o}}

We have already seen how Craciun and Feinberg's theorem can imply that of Thomas-Soul\'{e} in the case of CRNs.  Here, we show how Gale-Nikaid\^{o} can be used to prove the second half of Craciun-Feinberg.

Suppose that the cycle conditions of Theorem~\ref{CFlemma1} are met for a given CRN, implying that all of the determinants $\det([\mathbf{e_{j_1}}, \dots, \mathbf{e_{j_n}}])$ have the same sign.  This fact depends only on the network's structure and not at all on its rate constants.  Now, consider any subgraph of the interaction graph corresponding to a subset of $N$ of the species.  Any cycles in this subgraph are also present in the full graph, which means that the hypotheses of the theorem are still satisfied, and hence there is no way to generate coefficients with a sign of $(-1)^{N-1}$ in the principal minor corresponding to this subset.  (Because of the diagonal term, the minor is strictly nonzero.)  Thus, the matrix $\mathbf{J_{-F}}$ is a P-matrix, and by Gale-Nikaid\^{o}, the full network is injective.


This proof does not use Gale-Nikaid\^{o} to its full potential, since the hypotheses of Gale-Nikaid\^{o} do not require that minors be nonzero for all values of the rate constants.  Ideally, one would be able to find conditions, cycle-based or otherwise, that would give nonvanishing minors for a specific set or region of rate constants, even for a network that has some \textquotedblleft bad" cycles in it under Craciun-Feinberg.

\subsection{A note on reversible reactions and related issues}\label{tricky}

While the Thomas-Soul\'{e} theorem gives an elegant necessary condition for multistability, its power in the case of CRNs is limited by the prevalence of positive cycles.  Any reversible reaction, for example, even $A \leftrightharpoons B$, gives rise to a positive cycle, with positive edges from $A$ to $B$ and vice versa.  Regardless of the outflow properties, though, this reaction only has one equilibrium.  Even a single reaction $A + B \to C$, for that matter, yields a positive cycle, this time with negative edges between $A$ and $B$.  These \textquotedblleft trivial" examples might make us pessimistic about the applicability of any of our theorems, but, luckily, Craciun-Feinberg provides some relief. 

Let us assume degradation or outflow of all species, so that the Craciun-Feinberg results apply.  We have seen in section~\ref{CFsubsection1} that we can rearrange the expansion of $\det(\mathbf{J_F})$ such that the coefficient of each term is of the form $\det([\mathbf{c_{j_1}}, \dots, \mathbf{c_{j_n}}])\cdot \det([\mathbf{e_{j_1}}, \dots, \mathbf{e_{j_n}}])$.  As long as the signs of these coefficients are all the same, the CRN will be injective.  If we form a cycle with any repetition among the indices $j_1, j_2, \dots, j_N$, then the first of these determinants will be zero, since two rows will be identical.\footnote{In the context of their theorem, this issue never arises, due to the strict assignments of reaction indices $j_h$ (see section~\ref{CFsubsection1}).}  Thus we can declare that the $A + B \to C$ issue in the previous paragraph can be ignored.  

With regard to reversible reactions, Craciun and Feinberg combine both directions into a single reaction-node in their graphs, which eliminates the possibility of using a reaction and its reverse within the same cycle.  Within our framework, we can justify simply ignoring such cycles by noting that they lead to the presence of two rows in the second matrix above that differ by a factor of -1 throughout (and are thus linearly dependent).  Graphically, we can incorporate this exclusion by drawing single, two-sided arrows for reversible reactions in the interaction graph.

Of course, these trivial examples are not the only ways in which one of the determinant-coefficients above may vanish due to linear dependence.  All that is necessary is the alignment of a few stoichiometric coefficients.  In general, though, this will be hard to detect without checking individual cycles by hand.

As a final note, while Craciun and Feinberg's statment of their theorem eliminates these two forms of cycles, it simultaneously introduces at least two other potential trivialities.  First, their main theorem uses the non-oriented SR graph, some of whose cycles might not be true, oriented cycles.  Second, their method for drawing edges in the SR graph allows false edges between species that both appear on the product side of a given reaction, for example in the \textquotedblleft cycle" $A \Rightarrow B \Rightarrow A$ arising from the reaction $C \to A + B$.

\section{Multisite protein phosphorylation}\label{MAPK}

Perhaps the most-studied example of a bistable system in biochemistry is the eukaryotic cell-signaling module based on mitogen-activated protein kinase (MAPK) \cite{AS1, ASmulti, AFS, Bhalla, FX, JG1, HuangFerrell, Markevich}.  A \emph{kinase} is an enzyme that acts by attaching a phosphate group, usually taken from a molecule of ATP, to another molecule in order to effect a change in that molecule's chemical behavior.  An enzyme that removes a phosphate group is referred to as a \emph{phosphatase}. When abstracted, the key feature of the MAPK system is that MAPK is itself regulated by phosphorylation, and moreover, that it has two distinct phosphorylation loci.  Let us write $S_i$ for an $i$-times phosphorylated molecule of MAPK, for $i = 0, 1, 2$; $E$ and $F$ for the kinase and phosphatase that act on the $S_i$; and $ES_i$ and~$FS_i$ for the various intermediate enzyme-substrate complexes.  Part of the beauty of the MAPK system is that it can be accurately modeled using nothing more than repeated applications of Michaelis-Menten enzyme action \cite{Bhalla, JG1, HuangFerrell, Markevich}.  Thus, the full mechanism is the following.

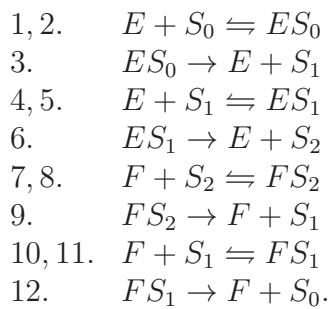
\[\begin{array}{ll}
1, 2. & E + S_0 \leftrightharpoons ES_0 \\
3. & ES_0 \to E + S_1 \\
4, 5. & E + S_1 \leftrightharpoons ES_1 \\
6. & ES_1 \to E + S_2 \\
7, 8. & F + S_2 \leftrightharpoons FS_2 \\
9. & FS_2 \to F + S_1 \\
10, 11. & F + S_1 \leftrightharpoons FS_1 \\
12. & FS_1 \to F + S_0. 
\end{array}\]

Many of the authors cited above have noted the connection between MAPK cascades and positive feedback, and indeed, we can see that the network contains a multitude of positive and sub-positive cycles, at all levels of the structure.  Each of the four phosphorylation and dephosphorylation events (reactions 1-3, 4-6, 7-9, and 10-12) contains its own two-element sub-positive cycle, for example $E \Rightarrow ES_0 \Rightarrow E$ using reactions 1 and 3.  We can also split the process into events at the first site (reactions 1-3 and 10-12) and those at the second site (4-9), and here we find the cycles $S_0 \Rightarrow ES_0 \Rightarrow S_1 \Rightarrow FS_1 \Rightarrow S_0$ (reactions 1, 3, 10, 12) and $S_1 \Rightarrow ES_1 \Rightarrow S_2 \Rightarrow FS_2 \Rightarrow S_1$ (reactions 4, 6, 7, 9).  These last two are significant because they strongly intersect with the sub-positive cycles $S_0 \Rightarrow E \Rightarrow S_1 \Rightarrow F \Rightarrow S_0$ (reactions 1, 4, 10, 12) and $S_2 \Rightarrow F \Rightarrow S_1 \Rightarrow ES_1 \Rightarrow S_2$ (reactions 7, 10, 4, 6), respectively.  While we cannot use Theorem~\ref{CFlemma1} to prove multistability based on these intersecting cycles, we can see, given that multistability is present, that they are responsible.

By contrast, consider a protein that is phosphorylated at only one site, with a mechanism consisting of reactions 1-3 and 10-12.  There are still positive cycles, so we cannot apply Thomas-Soul\'{e}, but now, the last two cycles listed above no longer exist, and in fact, there are no longer any strongly intersecting sub-positive cycles.  Thus, by Craciun-Feinberg, the one-site system is injective.

\section{Concluding remarks}

Multistability theory can be both delicate and forgiving.  Through an investigation of several approaches to the theory, we have seen that, just as multistable systems can be perturbed from one equilibrium state to another, slight differences in structure can push a network from the realm of the injective to the realm of the multistable.  At the same time, it is remarkable that the theorems we have examined, in particular that of Craciun and Feinberg, allow us to conclude so much from so little.



\medskip
\renewcommand\abstractname{Acknowledgements}
\begin{abstract}

I would like to thank Jeremy Gunawardena for generously supporting my work and supplying indispensable advice and feedback - both positive and negative as appropriate - throughout the process; Clifford Taubes for his practical guidance and his recommendations for the paper; and Adam Levine, Susan Lipson, and Rebecca Zeidel for their many useful comments on unfinished drafts.

\end{abstract}


\begin{thebibliography}{12}


\bibitem{Handwave} Aguda, Baltazar D. and Bruce L. Clarke.  \textquotedblleft Bistability in chemical reaction networks: Theory and application to the peroxidase-oxidase reaction."  \emph{Journal of Chemical Physics} \textbf{87:6} (15 September 1987), pp. 3461-3470.

\bibitem{AS1} Angeli, David and Eduardo Sontag.  \textquotedblleft Monotone control systems."  \emph{IEEE Transactions on Automatic Control} \textbf{48:10} (October 2003), pp. 185-202.

\bibitem{ASmulti} Angeli, David and Eduardo Sontag.  \textquotedblleft Multi-stability in monotone input/output systems."  \emph{Systems \& Control Letters} \textbf{51} (2004), pp. 1684-1698.

\bibitem{AFS} Angeli, David; Ferrell, James E., Jr.; and Eduardo Sontag.  \textquotedblleft Detection of multistability, bifurcations, and hysteresis in a large class of biological positive-feedback systems."  \emph{Proceedings of the National Academy of Sciences} \textbf{101:7} (17 February 2004), pp. 1822-1827.

\bibitem{Bailey} Bailey, James E.  \textquotedblleft Complex biology with no parameters."  \emph{Nature Biotechnology} \textbf{19} (June 2001), pp. 503-504.

\bibitem{BMQ} Bassett, Lowell; Maybe, John; and James Quirk.  \textquotedblleft Qualitative economics and the scope of the correspondence principle."  \emph{Econometrica} \textbf{36:3-4} (July-October 1968), pp. 544-563.

\bibitem{Bhalla} Bhalla, Upinder S. and Ravi Iyengar.  \textquotedblleft Emergent properties of networks of biological signaling pathways."  \emph{Science} \textbf{283} (15 January 1999), pp. 381-387.

\bibitem{BH} Briggs, G.E. and J.B.S. Haldane.  \textquotedblleft A note on the kinetics of enzyme action."  \emph{Biochemical Journal} \textbf{19} (1925), pp. 338-339.

\bibitem{Campbell} Campbell, L.A.  \textquotedblleft Rational Samuelson maps are univalent."  \emph{Journal of Pure and Applied Algebra} \textbf{92:3} (1994), pp. 227-240.

\bibitem{switch} Cherry, Joshua L. and Frederick R. Adler.  \textquotedblleft How to make a biological switch."  \emph{Journal of Theoretical Biology} \textbf{203} (2003), pp. 117-133.

\bibitem{Cinq1} Cinquin, Olivier and Jacques Demongeot.  \textquotedblleft Positive and negative feedback: striking a balance between necessary antagonists."  \emph{Journal of Theoretical Biology} \textbf{216} (2002), pp. 229-241.

\bibitem{ClarkeLong} Clarke, Bruce L.  \textquotedblleft Stability of complex reaction networks."  \emph{Advances in chemical physics} \textbf{42} (1980), pp. 1-213.

\bibitem{ClarkeShort} Clarke, Bruce L.  \textquotedblleft Complete set of steady states for the general stoichiometric dynamical system."  \emph{Journal of Chemical Physics} \textbf{75:10} (15 November 1981), pp. 4970-4979.

\bibitem{Eco} Cody, Martin L. and Jared M. Diamond, ed.  \emph{Ecology and Evolution of Communities}.  Cambridge, Mass: Belknap Press of Harvard University Press, 1975.

\bibitem{CF1} Craciun, Gheorghe and Martin Feinberg.  \textquotedblleft Multiple equilibria in complex chemical reaction networks: I. The injectivity property."  \emph{SIAM Journal of Applied Mathematics} \textbf{65:5} (2005), pp. 1526-1546.

\bibitem{CF2} Craciun, Gheorghe and Martin Feinberg.  \textquotedblleft Multiple equilibria in complex chemical reaction networks: II. The species-reaction graph."  \emph{SIAM Journal of Applied Mathematics} \textbf{66:4} (2006), pp. 1321-1338.

\bibitem{CF3} Craciun, Gheorghe and Martin Feinberg.  \textquotedblleft Multiple equilibria in complex chemical reaction networks: III. Extensions to entrapped species models." Mathematical Biosciences Institute Technical Report No. 36 (2005).

\bibitem{CFPNAS} Craciun, Gheorghe; Tang, Yangzhong; and Martin Feinberg.  \textquotedblleft Understanding bistability in complex enzyme-driven reaction networks."  \emph{Proceedings of the National Academy of Sciences} \textbf{103:23} (6 June 2006), pp. 8697-8702.

\bibitem{Degn} Degn, Hans.  \textquotedblleft Bistability caused by substrate inhibition of peroxidase in an open reaction system."  \emph{Nature} \textbf{217} (16 March 1968), pp. 1047-1050.

\bibitem{VDE} Essen, Arno van den.  \emph{Polynomial Automorphisms and the Jacobian Conjecture}.  Boston: Birkh\"{a}user Verlag, 2000.

\bibitem{Feinberg1} Feinberg, Martin.  \textquotedblleft Chemical reaction network structure and the stability of complex isothermal reactors: I. The deficiency zero and deficiency one theorems."  \emph{Chemical Engineering Science} \textbf{42:10} (1987), pp. 2229-2268.

\bibitem{Feinberg2} Feinberg, Martin.  \textquotedblleft Chemical reaction network structure and the stability of complex isothermal reactors: II. Multiple steady states for networks of deficiency one."  \emph{Chemical Engineering Science} \textbf{43:1} (1988), pp. 1-25.

\bibitem{FX} Ferrell, James E., Jr. and Wen Xiong.  \textquotedblleft Bistability in cell signaling: How to make continuous processes discontinuous, and reversible processes irreversible."  \emph{Chaos} \textbf{11:1} (March 2001), pp. 227-236.

\bibitem{GN} Gale, David and Hukukane Nikaid\^{o}.  \textquotedblleft The Jacobian matrix and global univalence of mappings."  \emph{Mathematische Annalen} \textbf{159} (1965), pp. 81-93.

\bibitem{Toggle} Gardner, Timothy S.; Cantor, Charles R.; and James J. Collins.  \textquotedblleft Construction of a genetic toggle switch in \emph{Escherichia coli}."  \emph{Nature} \textbf{403} (20 January 2000), pp. 1047-1050.

\bibitem{Gouze} Gouz\'{e}, Jean-Luc.  \textquotedblleft Positive and negative circuits in dynamical systems."  \emph{Journal of Biological Systems} \textbf{6:1} (1998), pp. 11-15.

\bibitem{JG1} Gunawardena, Jeremy.  \textquotedblleft Multisite protein phosphorylation makes a good threshold but can be a poor switch."  \emph{Proceedings of the National Academy of Sciences} \textbf{102:41} (11 October 2005), pp. 14617-14622.

\bibitem{HuangFerrell} Huang, Chi-Ying F. and James E. Ferrell, Jr.  \textquotedblleft Ultrasensitivity in the mitogen-activated protein kinase cascade."  \emph{Proceedings of the National Academy of Sciences} \textbf{93} (September 1996), pp. 10078-10083.

\bibitem{Hunt} Hunt, K. L. C. and P. M. Hunt.  \textquotedblleft Nonlinear dynamics and thermodynamics of chemical reactions far from equilibrium."  \emph{Annual Review of Physical Chemistry} \textbf{41} (1990), pp. 409-439.

\bibitem{Jelonek} Jelonek, Zbigniew.  \textquotedblleft Geometry of real polynomial mappings."  \emph{Mathematische Zeitschrift} \textbf{239} (2002), pp. 321-333.

\bibitem{Bose} Karmakar, Rajesh and Indrani Bose.  \textquotedblleft Graded and binary responses in stochastic gene expression."  \emph{Physical Biology} \textbf{1} (2004), pp. 197-204.

\bibitem{Orgo} Kasteleyn, Pieter Willem.  \textquotedblleft Graph theory and crystal physics," in \emph{Graph Theory and Theoretical Physics}, Frank Harary ed.  New York: Academic Press, 1967, pp. 44-110.

\bibitem{Immune} Kaufman, Marcelle and Ren\'{e} Thomas.  \textquotedblleft Model analysis of the bases of multistationarity in the humoral immune response."  \emph{Journal of Theoretical Biology} \textbf{129:2} (November 1987), pp. 141-162.

\bibitem{MathPhys}  Keener, James and James Sneyd.  \emph{Mathematical Physiology}.  New York: Springer, 1998.

\bibitem{Leib} Leib, T. M.; Rumschitzki, D.; and Martin Feinberg.  \textquotedblleft Multiple steady states in complex isothermal CFSTRs: I. General considerations."  \emph{Chemical Engineering Science} \textbf{43:2} (1988), pp. 321-328.

\bibitem{Markevich} Markevich, Nick I.; Hoek, Jan B.; and Boris N. Kholodenko.  \textquotedblleft Signaling switches and bistability arising from multisite phosphorylation in protein kinase cascades."  \emph{Journal of Cell Biology} \textbf{164:3} (2004), pp. 353-359.

\bibitem{Mcc1} McCuaig, William.  \textquotedblleft Even dicycles."  \emph{Journal of Graph Theory} \textbf{35:1} (2000), pp. 46-68.

\bibitem{Mcc2} McCuaig, William.  \textquotedblleft P\'{o}lya's permanent problem."  \emph{Electronic Journal of Combinatorics} \textbf{11} (2004) \#R79, 83 pp. (electronic).

\bibitem{LAC} Novick, Aaron and Milton Weiner.  \textquotedblleft Enzyme induction as an all-or-none phenomenon."  \emph{Proceedings of the National Academy of Sciences} \textbf{43:7} (15 July 1957), pp. 553-566.

\bibitem{Oda} Oda, Susumu and Ken-ichi Yoshida.  \textquotedblleft A short proof of the Jacobian conjecture in case of degree $\leq 2$."  \emph{Comptes Rendus Math\'{e}matiques de l'Acad\'{e}mie des Sciences (Canada)} \textbf{5:4} (1983), pp. 159-162.

\bibitem{Pinchuk} Pinchuk, Sergey.  \textquotedblleft A counterexample to the strong real Jacobian conjecture."  \emph{Mathematische Zeitschrift} \textbf{217} (1994), pp. 1-4.

\bibitem{PMO} Plahte, Erik; Mestl, Thomas; and Stig W. Omholt.  \textquotedblleft Feedback loops, stability and multistationarity in dynamical systems."  \emph{Journal of Biological Systems} \textbf{3:2} (1995), pp. 409-413.

\bibitem{Parth} Parthasarathy, T.  \emph{On Global Univalence Theorems}.  Berlin: Springer-Verlag, 1983.

\bibitem{Polya} P\'{o}lya, George.  \textquotedblleft Aufgabe 424."  \emph{Archiv der Mathematik und Physik} \textbf{20} (1913), p. 271.

\bibitem{Pom} Pomerening, Joseph R.; Sontag, Eduardo D.; and James E. Ferrell Jr.  \textquotedblleft Building a cell cycle oscillator: hysteresis and bistability in the activation of Cdc2."  \emph{Nature Cell Biology} \textbf{5:4} (April 2003), pp. 346-351.

\bibitem{Sam} Samuelson, Paul A.  \textquotedblleft Prices of factors and goods in general equilibrium."  \emph{Review of Economic Studies} \textbf{21:1} (1953), pp. 1-20.

\bibitem{SCL} Schlosser, Paul M. and Martin Feinberg.  \textquotedblleft A theory of multiple steady states in isothermal homogeneous CFSTRs with many reactions."  \emph{Chemical Engineering Science} \textbf{49:11} (1994), pp. 1749-1767.

\bibitem{Snoussi} Snoussi, El Houssine.  \textquotedblleft Necessary conditions for multistationarity and stable periodicity."  \emph{Journal of Biological Systems} \textbf{6:1} (1998), pp. 3-9.

\bibitem{Soule} Soul\'{e}, Christophe.  \textquotedblleft Graphic requirements for multistationarity."  \emph{ComPlexUs} \textbf{1} (2003), pp. 123-133.

\bibitem{Soule2} Soul\'{e}, Christophe.  \textquotedblleft Mathematical approaches to differentiation and gene regulation."  \emph{Comptes Rendus Biologies} \textbf{329} (2006), pp. 13-20.

\bibitem{Text} Strogatz, Steven H.  \emph{Nonlinear Dynamics and Chaos}.  Reading, Mass: Addison-Wesley, 1994.

\bibitem{Thomas1} Thomas, Ren\'{e}.  \textquotedblleft On the relation between the logical structure of systems and their ability to generate multiple steady states or sustained oscillations."  \emph{Springer Series in Synergetics} \textbf{9} (1981), pp. 180-193.

\bibitem{Wang} Wang, Stuart Sui Sheng.  \textquotedblleft A Jacobian criterion for separability."  \emph{Journal of Algebra} \textbf{65:2} (1980), pp. 453-494.


\end{thebibliography}
\end{document}